\documentclass[longbibliography,
 reprint,
 amsmath,amssymb,
 aps,
 pra,
]{revtex4-2}

\usepackage{graphicx}
\usepackage{dcolumn}
\usepackage{bm}
\usepackage{braket}
\usepackage{color}
\usepackage{amsthm}
\newtheorem{thm}{theorem}
\newtheorem{lemma}[thm]{Lemma}
\newtheorem{theorem}[thm]{Theorem}

\begin{document}

\preprint{APS/123-QED}

\title{Topology-defined computation in knitted textiles}

\author{Daisuke S. Shimamoto}
\affiliation{
Graduate School of Arts and Sciences, The University of Tokyo, 3-8-1 Komaba, Meguro, Tokyo, 153-8902, Japan
}
\affiliation{Research Organization of Science and Technology, Ritsumeikan University, 1-1-1 Noji-higashi, Kusatsu, Shiga 525-8577, Japan}

\date{\today}

\begin{abstract}
Mechanical computation, in which logic functions are realized through deformation rather than electronics, has been demonstrated in systems such as origami, kirigami, and mechanical metamaterials. In these systems, logic states and functions are typically determined by geometry and material properties, making it sensitive to deformation and imperfections. Here we introduce a mechanical computing architecture in which logic is defined by topology rather than geometry. The circuit is realized as a knitted textile formed from a single continuous yarn, where information is encoded in the topology of stitches and processed through controlled unraveling. By discretizing the textile into a lattice of interacting cells, we construct topological propagation rules that implement universal logic operations, including NOT, AND, and OR gates, as well as a half-adder. Experiments demonstrate that the logical output is robust against geometric deformation, while mechanical factors affect only if the computation can be executed. These results establish topology-defined computation as a model for information processing in textiles and other reconfigurable physical systems.

\end{abstract}

\maketitle

\section{introduction}

Mechanical computing refers to information processing carried out in the materials through the propagation of motions, deformations, and transitions between multiple structures~\cite{yasuda2021mechanical}.
While the electronic computer surpasses it in compactness and speed, research on mechanical computers has recently been gathering interest due to their independence on power supply, adaptability to soft robots, and resilience to perturbation from the environments such as heat, electromagnetic waves, and radiation.
Recent advances have enabled many kinds of sophisticated mechanical computers, based on origami~\cite{treml2018origami,meng2021bistability} and kirigami~\cite{yang2024mechanical} structures, soft pneumatic and fluidic devices~\cite{weaver2010static,mosadegh2011next,preston2019digital,song2021cmos,stanley2024high,drotman2021electronics}, and elastic materials~\cite{raney2016stable,el2022mechanical,byun2024integrated,kwakernaak2023counting,waheed2020boolean,jiao2023mechanical,mei2023memory,mousa2024parallel,michel2025model,bilal2017bistable,zhang2023mechanical,merkle1993two}. These systems exploit nonlinear mechanical response and engineered multistability to realize binary states, implement logic gates, and transmit signals. Advances in fabrication processes have enabled materials to be sculpted into complex shapes with finely tuned energy landscapes, realizing mechanical computers with intricate designed responses. Such devices have been explored for applications in autonomous soft machines~\cite{drotman2021electronics}, adaptive structures, and environment-responsive materials.

Despite these advances, mechanical computers face a major limitation. They rely on fine-tuned geometry, and hence are sensitive to perturbations. Because logic bits are typically encoded as local energy minima, and computation corresponds to transitions between these minima, the behavior of the devices depends on its potential energy landscape sensitively. Small variations in geometry, stiffness, or boundary conditions can modify the mechanical response and potentially affect computational performance. The reliance on finely engineered geometry not only complicates fabrication but also limits robustness under deformations, such as  stretching, bending, or shear. Developing a mechanical computer whose logic function is robust to geometric and mechanical fluctuations remains an open challenge.

Here, we show that topology defines computation whereas mechanics only executes it in a knitted fabric. Computation is realized through controlled unraveling designated by the topology of the textile. The logic behavior is defined by topology, instead of by geometry or energy landscapes.

Knitted fabrics are traditional crafts~\cite{emery1980principle,rutt1987history,seiler1994textiles}, and are attractive materials for topology-defined computation because they can be fabricated by both automated machines and manual techniques, and their manufacturing methods are well established.
Recently, it have been also developping as mechanical metamaterials with advenced understandings and techniques~\cite{poincloux2018geometry,gonzalez2024pulling,mahadevan2024knitting,knittel2020modelling,storck2022topology,singal2024programming,niu2025geometric,maziz2017knitting,abel2013hierarchical}, where stitch topology governs elasticity and nonlinear deformation~\cite{liu2017role,poincloux2018geometry,singal2024programming,ding2024unravelling,tajiri2025curling}. A distinct feature of knitting is laddering (Fig.~\ref{ladder_fig}), a chain of unraveling, which is a topologically governed propagation process~\cite{g565-3dyn} and recently gathering interests of physicists~\cite{liu2026composite,faulconnier2026laddering}. We construct transmission channels and logic gates whose function depends solely on the interlinking of a single continuous yarn. Inputs and outputs are encoded by the knitted/unknitted state of designated stitches, and computation is executed through unraveling which propagates over the textile. 

This architecture establishes a clear separation between logic and mechanics. The function is defined by yarn topology, while the material supplies the media and forces needed to execute unraveling. Because the state of each stitch is discretized by topology rather than by energetic minima separated by finite potential barrier, the logical behavior becomes less sensitive to geometric variations and mechanical perturbations.

\begin{figure}
    \centering
    \includegraphics[width=0.8\linewidth,bb=0 0 290 290]{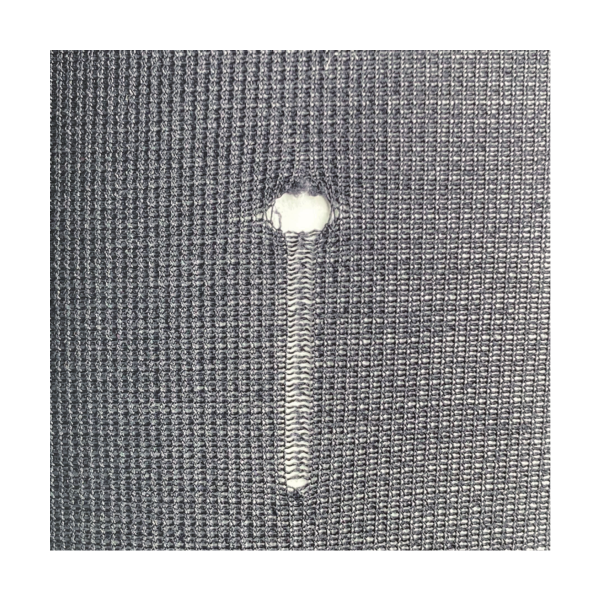}
    \caption{Laddering is caused by releasing a single loop, and it propagates over a knitted fabric by unraveling. The laddering is utilized as a information transmission in the present work.}
    \label{ladder_fig}
\end{figure}

\section{Discretized description of unraveling}
\label{theory}

Among variety of textiles, knitting exhibits a unique phenomenon known as laddering. Laddering is a chain of local failures that propagates through the fabric when a loop is released (Fig.~\ref{ladder_fig}), and it has been regarded as an undesirable breakdown in textile.

However, we utilize it for information processing by regarding it as an transmission of signals. The input and output information is represented by interlinking topology between loops in a fabric, and the information is transmitted and processed by unraveling.
Description and control of the laddering by topology are achieved in the periodic knitting by authors~\cite{g565-3dyn}.
We here establish the method for information processing by expanding it to knitting without periodicity. 

The unraveling is described as a dynamics on the lattice.
Our method is based on discretizing the states of the stitches into knitted and unknitted states and propagation of unknitted state.
The knittings are composed of cells aligned on the square lattice (Fig.~\ref{uc}).
Each cell, $l$ corresponds to a stitch in the textile and has an address specified by two integers. Cells have discretized physical states $\sigma_l\in\{K,U\}$, where $K$ denotes a knitted stitch and $U$ denotes an unknitted one (Fig.~\ref{uc}(a)). Note that we say that a cell is unraveled if all crossings within it can be removed only by continuous deformations without cutting or pasting the yarn as shown in the third step in the example in Fig.~\ref{uc}(b). In the example shown in Fig.~\ref{uc}(a), the cell $(1,-1)$ has a state $\sigma_{(1,-1)}=U$, and the next one has a state $\sigma_{(1,-2)}=K$.

We introduce the \textit{propagation} of a unknitted stitch. Propagation is defined as the transition of a stitch from the knitted state to the unknitted state induced by neighboring unknitted stitches. More precisely, a stitch is said to be released if there exists a continuous deformation of the yarn that transforms the stitch into the unknitted state while keeping the topology of the remaining textile unchanged, assuming that a specified set of neighboring cells is already unknitted. The update to change the released cell from $U$ to $K$ is called propagation.
In the example shown, the cell just below the cell with the state $U$ can also become $U$ via continuous deformations (Fig.~\ref{uc}(b)), which is what we call propagation.

It is an inherent property of knitting that each stitch has two states, and the mapping from physical states to a logic value is defined based on this two-state description. The spatial discretization into unit cells and state discretization into $U$ or $K$ enables unified description and control of the unraveling for many kinds of knittings.
Laddering proceeds through the unraveling of the stitches and the transitions of the cells from $K$ to $U$. Laddering is complex and difficult to describe as dynamics of yarn segments, but it can be easily handled through discretization~\cite{g565-3dyn}. 

We describe the programmed way of unraveling as dependencies among the cells. For each cell $l$, we specify a \textit{propagation set}
\begin{equation}
S_l=\{R_{l,1},R_{l,2},\ldots\},
\end{equation}
where each prerequisite set $R_{l,k}$ lists the cells that must already be in state $U$ before cell $i$ is allowed to transition from $K$ to $U$. Each element of $R_{l,k}$ is represented as a relative address from the cell $l$. By making a cell with a topology with desired propagation set, the way of unraveling can be controlled.

By descretizing the time into steps, $t$, and considering the unraveling process as repeated applications of the update operation, we can formally describe the unraveling process.
Letting $\sigma_l^{(t)}\in\{K,U\}$ denote the state of the cell $l$ at step $t$ and $V=\{1,\dots,N\}$ be the set of cells, the update is that $\sigma_l^{(t+1)}=U$ if $\sigma_l^{(t)}=U$ or there exists $k$ with $R_{l,k}\subseteq\{j\in V|\sigma_{l+j}^{(t)}=U\}$, and otherwise, $\sigma_l^{(t+1)}=K$. Here, we denote by the sum between cells the sum of each coordinate.

\begin{figure}
    \centering
    \includegraphics[width=1.0\linewidth,bb=0 0 260 330]{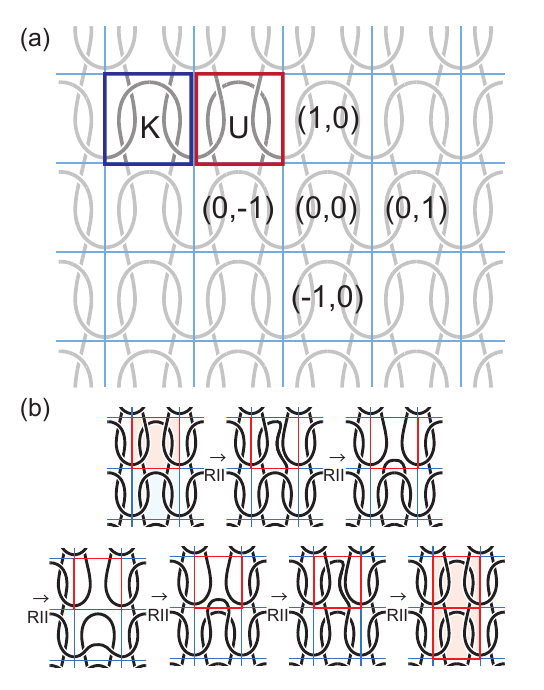}
    \caption{(a) A schematic of a knitted logic curcuit. Each cell has an address specified with an ordered pair of integers. In this example, a pair of cells $((1,-2),(1,-1))$ represents logic value of 1, because $\sigma_{(1,-2)}=K$ and $\sigma_{(1,-1)}=U$. (b) Continuous deformation can remove all crossings within the unknitted cell, and further deformation can duplicate the unknitted cell, which we call propagation.}
    \label{uc}
\end{figure}

Following this description, knitted fabrics are regarded as models on the lattice, such as cellular automata or kinetically constrained models. For example, all the stitches in jersey knitting has a propagation set $\{\{(1,0)\}\}$, and simple crochet has a propagation set $\{\{(0,1)\}\}$. As a models on the lattice, both of them correspond to the East model~\cite{jackle1991hierarchically,sollich1999glassy}.

Given an initial configuration $\sigma^{(0)}$ encoding the input bits, successive applications of the update define a propagation process.
The outcome is read from the final states at designated output cells. 
Geometric parameters such as tension or elasticity may affect the rate and ease of transitions, but the final state is prescribed solely by the propagation sets $S_i$, i.e., by the yarn topology.

\section{Uniqueness of the computation result}
\label{app_uni}

The information is input and output stitch-wise, and the transmission and processing is executed through unraveling stitch by stitch. The input and output stitches are assigned at first, and some input stitches are unknitted to input bits and execute the computation. By observing which output stitches have been unknitted after repeted update, the output is read out. We here show that the result is unique and independent of the order of the unraveling events.

This computation process can be regarded as a dynamics on the lattice as well as monotone cellular automata~\cite{morris2017bootstrap} or boolean networks~\cite{bollobas2015monotone}.
Given an initial configuration $\sigma^{(0)}$ encoding the input bits, successive applications of the update define a propagation process.

Each cell $i\in V$ has a physical state $\sigma_i\in\{K,U\}$ with $K<U$.
The configuration space $\mathcal{X}:=\{K,U\}^N$ is a set of states with the partial order,
\begin{equation}
\sigma\le\tau\, :\Leftrightarrow\, ^\forall i\in V,\; \sigma_i\le \tau_i,
\end{equation}
where $\sigma$, $\tau\in\mathcal{X}$.
Thus $\sigma\le\tau$ means that $\tau$ has all cells unknitted that are unknitted in $\sigma$, and possibly more are unknitted.

For each cell $i$, the local propagation rule is described as a set of prerequisite sets
\begin{equation}
S_i=\{R_{i,1},\dots,R_{i,m_i}\},\qquad R_{i,k}\subseteq V,
\end{equation}
meaning that $i$ may change from $K$ to $U$ if all cells in some $R_{i,k}$ are $U$.
Define the global update operator $\Phi:\mathcal{X}\to\mathcal{X}$ cellwise by
\begin{equation}
\Phi_i(\sigma)=
\begin{cases}
U,\;\; \text{if }\sigma_i=U\text{ or }\sigma_i\text{ is released},\\
K,\;\; \text{otherwise,}
\end{cases}
\end{equation}
where we call $\sigma$ is released if the following condition is satisfied:
\begin{equation}
    ^\exists k\in\{1,\dots,m_i\},\; ^\forall j\in R_{i,k}: \sigma_j=U.
\end{equation}
Equivalently, $\Phi(\sigma)=\sigma\vee R(\sigma)$, where $R_i(\sigma)=U$ if and only if $\sigma_i$ is released. The join $\vee$ denotes cellwise join in the finite lattice $\mathcal{X}$, which join means a cell is $U$ if eather of them is $U$ and $K$ otherwise.

\begin{lemma}
\label{tancho}
The map $\Phi$ is monotone and inflationary:
\begin{equation}
\sigma\le\tau\;\Rightarrow\;\Phi(\sigma)\le \Phi(\tau),
\;
\sigma\le \Phi(\sigma)\quad(^\forall \sigma\in\mathcal{X}).
\end{equation}
\end{lemma}

\begin{proof}
If $\sigma\le\tau$ and $\sigma_i$ is released, then $\tau_j=U$ for those $j$ such that $\sigma_j=U$, hence $\tau_i$ is released as well.
Thus $R(\sigma)\le R(\tau)$ componentwise.
Since $\Phi(\cdot)=\cdot\vee R(\cdot)$ and join is monotone, $\Phi$ is monotone.
Inflationarity follows from $\sigma\le \sigma\vee R(\sigma)=\Phi(\sigma)$.
\end{proof}

Consider the iteration starting from a given initial configuration $\sigma^{(0)}$:
\begin{equation}
\label{iteration}
\sigma^{(t+1)}=\Phi\bigl(\sigma^{(t)}\bigr),\qquad t=0,1,2,\dots
\end{equation}

\begin{lemma}
\label{fixed}
The sequence $\sigma^{(0)}\le \sigma^{(1)}\le \sigma^{(2)}\le\cdots$ defined by Eq. \ref{iteration} goes to a state $\sigma^*$ satisfying $\Phi(\sigma^*)=\sigma^*$ in a finite number of update steps.
Moreover, $\sigma^*$ is the least fixed point of $\Phi$ for $\sigma^{(0)}$, \emph{i.e.}, 
for any $\tau$ with $\sigma^{(0)}\le\tau=\Phi(\tau)$, $\sigma^*\le\tau$.
\end{lemma}

\begin{proof}
By Lemma~\ref{tancho}, the sequence is nondecreasing.
Each strict increase flips at least one cell from $K$ to $U$, which can happen at most $N$ times.
Hence the chain stabilizes to a $\sigma^*$ with $\Phi(\sigma^*)=\sigma^*$ in a finite number of steps.
For leastness, if $\tau$ is any fixed point with $\sigma^{(0)}\le\tau$, then by monotonicity $\sigma^{(t)}\le \tau$ for all $t$, hence $\sigma^*\le\tau$.
\end{proof}

A single physical step of unknitting is next modelled in a formal way. It may update any subset of cells that are released at that moment, possibly just one, or several.
Let us call a one-step map $U:\mathcal{X}\to\mathcal{X}$ \emph{possible} if for all $\sigma\in\mathcal{X}$ it satisfies:
\begin{itemize}
    \item[(i)]$\sigma\le U(\sigma)\le\Phi(\sigma)$,
    \item[(ii)] If $\Phi(\sigma)>\sigma$, then $U(\sigma)>\sigma$.
\end{itemize}
The left inequality in (i) comes from the irreversibility (from $K$ to $U$ only), and the right inequality encodes that no step can exceed the parallel update. The condition (ii) implies that unraveling proceeds whenever any cell is released.
A possibly varying sequence of possible steps $\{U_t\}$, and equivalently the sequence of states $\tilde\sigma^{(t)}$, define a physical trajectory
\begin{equation}
\label{keiro}
\tilde\sigma^{(t+1)} = U_t\left(\tilde\sigma^{(t)}\right),\quad\tilde\sigma^{(0)}=\sigma^{(0)}.
\end{equation}

\begin{lemma}
\label{samefix}
Let $U:\mathcal{X}\to\mathcal{X}$ be a possible one-step map, then for every state $\sigma\in\mathcal{X}$,
\begin{equation}
U(\sigma)=\sigma \Leftrightarrow \Phi(\sigma)=\sigma.
\end{equation}
\end{lemma}

\begin{proof}
Suppose $U(\sigma)=\sigma$.  
If $\Phi(\sigma)>\sigma$, the assumption (ii) require $U(\sigma)>\sigma$, contradicting $U(\sigma)=\sigma$.  
Hence $\Phi(\sigma)=\sigma$.

Conversely, if $\Phi(\sigma)=\sigma$, then $\sigma\le U(\sigma)\le \Phi(\sigma)=\sigma$, so $U(\sigma)=\sigma$.
\end{proof}

\begin{lemma}
For every possible trajectory~(\ref{keiro}), the sequence
$\tilde\sigma^{(0)}\le \tilde\sigma^{(1)}\le\cdots$ stabilizes in a finite number of steps to the same fixed point $\sigma^*$ as in Lemma~\ref{fixed}, independently of the choice of possible steps.
\end{lemma}

\begin{proof}
From condition (i), $\tilde\sigma^{(t)}\le \tilde\sigma^{(t+1)}$, and whenever $\Phi(\tilde\sigma^{(t)})>\tilde\sigma^{(t)}$ at least one cell flips, so the chain stabilizes in at most $N$ strict increases.
By induction on $t$ and and monotonicity of $\Phi$,
\begin{equation}
    \tilde\sigma^{(t+1)}=U_t(\tilde\sigma^{(t)})\le \Phi\bigl(\tilde\sigma^{(t)}\bigr)\le \Phi^{t+1}\bigl(\sigma^{(0)}\bigr),
\end{equation}
hence $\tilde\sigma^{(t)}\le \Phi^t(\sigma^{(0)})$ for all $t$, and $\lim_{t\rightarrow\infty} \tilde\sigma^{(t)}\le \sigma^*$.

Moreover, if $\tilde\sigma^{(t)}=\tilde\sigma^{(t+1)}$, then $U_t\bigl(\tilde\sigma^{(t)}\bigr)=\tilde\sigma^{(t)}$. By Lemma~\ref{samefix}, $\Phi(\tilde\sigma^{(t)})=\tilde\sigma^{(t)}$.
Since the limit $\sigma'$ of the possible chain is a fixed point of $\Phi$ with $\sigma^{(0)}\le\sigma'$, Lemma~\ref{fixed} yields $\sigma^*\le\sigma'$. Together with the upper bound $\sigma'\le\sigma^*$, $\sigma'=\sigma^*$.
This limit is independent of the sequence $\{U_t\}$.
\end{proof}

\begin{theorem}
Given the initial configuration $\sigma^{(0)}$ (encoding the inputs via the two-cell convention), any physically possible unraveling process, regardless of the order or simultaneity of unraveling events, converges in finitely many steps to the same terminal configuration $\sigma^*$ characterized by $\Phi(\sigma^*)=\sigma^*$ and $\sigma^{(0)}\le\sigma^*$.
The logical output, read from which side of each output stitch finally unravels, is therefore unique and independent of update order.
\end{theorem}

\section{Propagation sets for logic}
\label{propagation_sets_for_logic}

Propagation sets provide a way to construct logic circuits. We here demonstrate the construction of logic bit, transmission channels, and logic gates.

\subsection{Representation of bits}

A logic bit is represented by an ordered pair of adjacent cells. The right and left cells are denoted by subscripts, $l$ and $r$, respectively.
If $(\sigma_l,\sigma_r)=(U,K)$, it represents 0, and if $(\sigma_l,\sigma_r)=(K,U)$, it represents 1. The configuration $(K,K)$ is used as a neutral state. An example is shown as Fig.~\ref{gates}(a). The examples throughout this paper is based on stockinette knitting although this computer can be constructed on other types of knitted or crocheted fabrics, because it is the easiest to make and the simplest to show.
With this two-cell convention, every input and output has a unique physical realization, and the output is read from which side unknitted in the pair.

\begin{figure*}
    \centering
    \includegraphics[width=1.0\linewidth,bb=0 0 540 250]{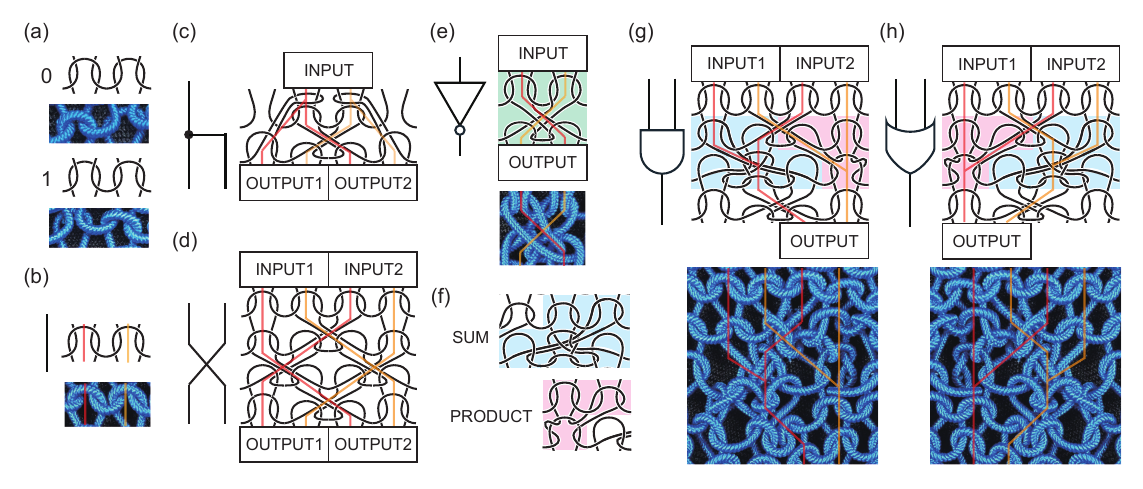}
    \caption{Each logic components are shown as a set of a diagram in a two-dimensional projection and an experimental photo. (a) Logic bit 0 and 1, (b) channel, (c) FANOUT, (d) EXCHANGE, (e) NOT gate, (f) sum and product modules, (g) AND gate, and (h) OR gate, which is a mirror image of AND gate. In this figure, the yarn segments are fixed by pins to prevent the defect propagation and improve visibility.}
    \label{gates}
\end{figure*}

\subsection{Transmission}

The information is transmitted through chain of unraveling in the knitted fabrics.
Reliable transmission and routing of signals are essential for scalable computation.
In the knitted computer, transmission can be realized by the topology of the yarn.

Three basic transmission elements are used in the logic circuit: the channel, FANOUT, and EXCHANGE, as illustrated in Fig.~\ref{gates}(b-d).
Each element is defined in terms of the relationships between the left and right cells of the input and output bits. The functions of the elements are described by the propagation set introduced in the previous section.

A channel simply transmit the information.
A single input bit, $A$ transmits its logic state to a single downstream bit, $C$. 
The propagation sets are
\begin{equation*}
S_{C_l} = \big\{\{A_l\}\big\}, S_{C_r} = \big\{\{A_r\}\big\},
\end{equation*}
so that the unraveling of $A$ directly triggers that of $C$.
This corresponds to a straight transmission. An example of this element with $S_{C_l}=\big\{\{(1,0)\}\big\}$ and $S_{C_r}=\big\{\{(1,0)\}\big\}$ is shown in Fig.~\ref{gates}(b).

A FANOUT duplicate the input signal and transmit it to multiple outputs.
The input bit $A$ transmits its state simultaneously to two downstream bits $C$ and $D$. Each receives the same logic value.
The propagation sets are
\begin{eqnarray*}
S_{C_l} = \big\{\{A_l\}\big\}, S_{C_r} = \big\{\{A_r\}\big\},\\ S_{D_l} = \big\{\{A_l\}\big\}, S_{D_r} = \big\{\{A_r\}\big\}.
\end{eqnarray*}
An example with $S_{C_l}=\big\{\{(3,1)\}\big\}$, $S_{C_r}=\big\{\{(3,0)\}\big\}$, $S_{D_l}=\big\{\{(3,0)\}\big\}$, and $S_{D_r}=\big\{\{(3,-1)\}\big\}$ is shown in Fig.~\ref{gates}(c).
The unraveling initiated at $A$ thus branches into two directions, realizing FANOUT.

By EXCHANGE, two input bits, $A$ and $B$ exchange their positions as they propagate downstream, producing outputs, $D$ and $C$, respectively:
\begin{eqnarray*}
S_{C_l} = \big\{\{B_l\}\big\}, S_{C_r} = \big\{\{B_r\}\big\},\\
S_{D_l} = \big\{\{A_l\}\big\}, S_{D_r} = \big\{\{A_r\}\big\}.
\end{eqnarray*}
An example with $S_{C_l}=\big\{\{(5,2)\}\big\}$, $S_{C_r}=\big\{\{(5,2)\}\big\}$, $S_{D_l}=\big\{\{(5,-2)\}\big\}$, and $S_{D_r}=\big\{\{(5,-2)\}\big\}$ is shown in Fig.~\ref{gates}(d).
This operation corresponds to a crossing of two channels, in which two unraveling paths cross so that each of them passes behind or in front of the other without interaction.

These three elements provide all necessary routing functions for constructing complex circuits.

\subsection{Processing}
We construct NOT, AND, and OR gates, which form a functionally complete set, to show that any logic function can be realized by a single yarn. 
A NOT gate is realized by interchanging the left and right order of a bit pair, so that the unraveling that would designate 0 becomes the unraveling that designates 1, and vice versa. Accordingly, for input $A=(A_l,A_r)$ and output $C=(C_l,C_r)$, NOT gate is yielded by setting $S_{C_l}=\{\{A_r\}\}$ and  $S_{C_r}=\{\{A_l\}\}$. See an example with $S_{C_l}=\big\{\{(1,1)\}\big\}$ and $S_{C_r}=\big\{\{(1,-1)\}\big\}$ in Fig.~\ref{gates}(e).

We next introduce AND and OR gates. They are realized by wiring the left and right cells of input pairs to the left and right cells of an output pair through elementary modules. We first introduce the two modules to simplify the description of the gates.
Both the modules have two input cells and one output cell. 
The first one is \emph{sum} module. 
For a given input to the two upstream cells, the downstream cell is permitted to unknit if either of the two upstream cells has unknitted. 
If $a$ and $b$ denote the two upstream cells and $c$ the downstream cell, the propagation set is
\begin{equation}
S_c=\big\{\{a\},\{b\}\big\}.
\end{equation}
This module transmits the presence of unraveling from either input.

The second one is \emph{product} module.
The downstream cell $c$ remains knitted until both upstream cells have unknitted:
\begin{equation}
S_c=\big\{\{a,b\}\big\}.
\end{equation}
This module requires concurrence of unknitting at both inputs.
An example for each modules with $S=\big\{\{(1,0)\},\{(1,1)\}\big\}$ and $S=\big\{\{(1,0), (1,1)\}\big\}$ are shown in Fig.~\ref{gates}(f).
Using these modules, two-input Boolean gates on bit pairs, AND and OR, are assembled as follows.

AND gate utilizes both modules. For inputs $A=(A_l,A_r)$ and $B=(B_l,B_r)$ and output $C=(C_l,C_r)$,
the propagation sets are set as $S_{C_l}=\big\{\{A_l\},\{B_l\}\big\}$ and $S_{C_r}=\big\{\{A_r,B_r\}\big\}$, which means one 
connects the left cells through a product module so that $C_l$ unknits only if both $A_l$ and $B_l$ have unknitted, and
connects the right cells through a sum module so that $C_r$ unravels if either $A_r$ or $B_r$ has unknitted. See an example with $S_{C_l}=\big\{\{(5,-2)\},\{(5,0)\}\big\}$ and $S_{C_r}=\big\{\{(5,-2),(5,0)\}\big\}$ in Fig.~\ref{gates}(g).
Under the two-cell readout rule, this mapping realizes $C = A \land B$.
Interchanging $C_l$ and $C_r$ yields a NAND.

Reversing the assignment used for AND yields OR gate.
The propagation sets are set as $S_{C_l}=\big\{\{A_l\},\{B_l\}\big\}$ and $S_{C_l}=\big\{\{A_r,B_r\}\big\}$. It implies that the mirror image of OR yields AND gate. See an example with $S_{C_l}=\big\{\{(5,0),(5,2)\}\big\}$ and $S_{C_r}=\big\{\{(5,0)\},\{(5,2)\}\big\}$ in Fig.~\ref{gates}(h).
The pairwise readout then realizes $C = A \lor B$.

In this section, we showed the propagation sets to construct a computer in a knitted fabrics. It can be interpreted as site-dependent monotone Boolean functions. Unlike conventional monotone Boolean networks, these functions are not prescribed abstractly but emerge from the topology within each cell. We are able to experimentally implement the dynamics on the networks.

\begin{figure*}
    \centering
    \includegraphics[width=1.0\linewidth,bb=0 0 520 340]{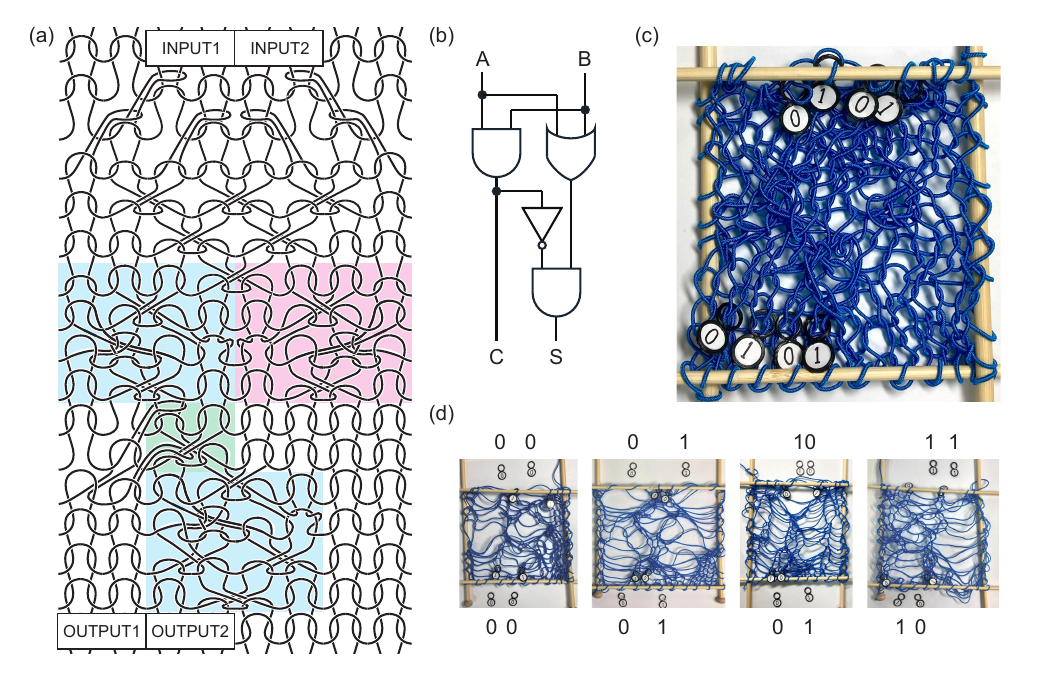}
    \caption{(a) Diagram of knitted half-adder realized in our experiment. It adds two input bits, A and B, to produce sum bit, S and carry bit, C. The logic gates of NOT, AND, and OR are highlighted in green, blue, and red, respectively. (b) The circuit diagram of the half-adder illustrated in (a). (c) Experimental realization of the half-adder. The tags of the input and output bits are attatched to the yarn loop for a visibility. The tags correspond to the input and output logic values are released from the computer. (d) The half-adders after the computations for all four inputs, $A=0$ and $B=0$, $A=0$ and $B=1$, $A=1$ and $B=0$, and $A=1$ and $B=1$. All outputs are confirmed to be correct.}
    \label{HA}
\end{figure*}

\section{Experimental verification of logic gates}

\subsection{Fundamental logic gates}

The elementary logic gates NOT, AND, and OR were realized experimentally, as shown in Fig.~\ref{gates}. For all tested input combinations, the resulting outputs agreed with the corresponding truth tables. Although unraveling events occurred asynchronously, the final state was independent of their microscopic order, indicating that the output is determined by the propagation rules rather than by the details of the dynamics. These results demonstrate that the fundamental modules correctly implement the logic operations required for computation.

\subsection{Half-adder circuit}

To verify that correct logical behavior is preserved under gate composition and our construction is scalable, we implemented a half-adder circuit constructed from the elementary gates. 
The design, shown schematically in Fig.~\ref{HA}(a,b), consisted of an AND gate for the carry output and an XOR pathway, which is built from an AND, OR, and NOT gate, for the sum output. The entire structure was knitted from a single yarn (Fig.~\ref{HA}(c)).

For each of the four input combinations $(A,B)\in\{0,1\}^2$, computation was initiated by unknitting the input stitches corresponding to the specified bits.
Propagation proceeded through the two parallel pathways delivering signals to the carry and sum outputs.
The final outputs were read by observing which stitch unraveled.
Small plastic rings were attached to the stitches representing input and output bits to aid visibility.
These markers did not participate in the computation, and they served only as visual tags for identifying the logic state of bit pairs.
The result was easily read by observing which tagged rings detached from the output pairs.
Figure~\ref{HA}(d) presents the experimentally obtained final configurations for all inputs.
In all four cases, the outputs matched the expected truth table.
This multi-gate circuit demonstrates that topologically defined logic remains consistent when gates are combined, and that the physical unraveling executes the logical dependencies determined by the topology.

In addition to universal logic, the same method can realize monotone Boolean functions using a more compact single-cell representation of logic bits. Because such circuits lose universality instead of compactness, we present them in Appendix~\ref{monotone_section}.

\subsection{Robustness under geometric deformation}

The experiment was done also after twisting, stretching, or compressing the fabric.
To test the invariance of logical behavior under geometric perturbations, the same knitted samples were subjected to twisting, stretching, and compression before actuation.
In all cases, unraveling propagated along the same pathways and terminated in identical outputs.
Even when the fabric was twisted 180$^\circ$ about its center, bended to fold or stretched by more than 50\% in one direction, the pattern of unraveling was unchanged (Fig.~\ref{robust}).
This confirms that this computer is robust to geometric deformations.

\begin{figure}
    \centering
    \includegraphics[width=1.0\linewidth,bb=0 0 260 230]{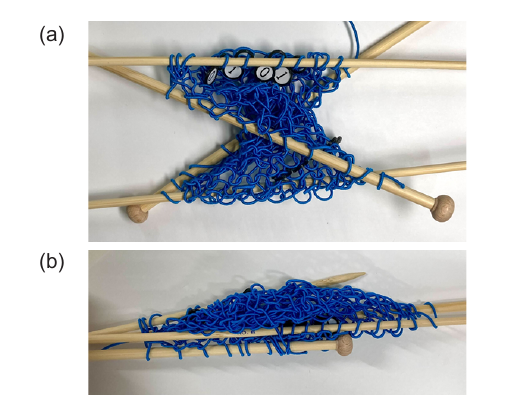}
    \caption{It was confirmed that the output was robust against (a) twisting and (b) bending.}
    \label{robust}
\end{figure}

\subsection{Role of mechanics}

\begin{figure}
    \centering
    \includegraphics[width=0.9\linewidth,bb=0 0 170 210]{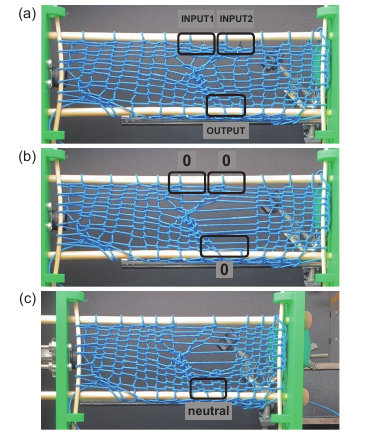}
    \caption{(a) A knitted AND gate stretched under a tension of 21 N. Two stitches are held by metal rings to initialize the input state. (b) Computation after applying the input (0, 0) by removing the rings. The unraveling propagated through the textile and produced the expected output. (c) The same gate after reknitting the unraveled stitches and applying a smaller tension of 18 N. In this case, the unraveling stopped before reaching the output stitch, resulting in an incomplete computation. No incorrect output is produced, and the incomplete state is distinguished from a completed computation.}
    \label{mechanics}
\end{figure}

Although the present work focuses on topology rather than mechanics, it is useful to clarify the role of mechanics in the execution of computation. The logic introduced in this study is defined by the topology of the textile and are independent of the particular yarn material or geometry. Nevertheless, the physical execution of the computation requires the propagation of unraveling, which is influenced by mechanical factors such as tension, friction, and yarn properties. Previous studies have shown that the mechanical response of knitted fabrics depends on both constituent materials and manufacturing parameters~\cite{gonzalez2024pulling}. Therefore, while mechanics does not determine the logic, it can influence if and how efficiently the computation is carried out.

Insufficient mechanical driving does not generate incorrect logical outputs. Instead, it prevents the computation from reaching completion. Figure~\ref{mechanics} shows experiments of logic gate performed under different stretching forces, where a knitted gate was first stretched to a designated force and then implemented by releasing selected stitches corresponding to the input bits. In the example shown, the computation completed successfully at 21~N, producing the expected output, whereas at 18~N the unraveling stopped before reaching the output bit. Importantly, the incomplete case was distinguishable from a completed computation because both the output stitches remained knitted. We thus observed either successful computation or incomplete propagation, but not an incorrect logical output.

These observations suggest that mechanics determines if a computation can be executed, and the logical output is determined by topology. The threshold force required for successful propagation is not expected to be universal because the mechanical response of knitted fabrics depends on yarn properties, geometry, and frictions. But, this separation between logic and mechanics is a key feature of the knitted computer.

\section{Discussion}
\label{discussion}

Mechanical computation has been realized in a wide variety of physical systems, including kirigami and origami structures~\cite{yang2024mechanical,treml2018origami}, architected elastic materials~\cite{raney2016stable,byun2024integrated,kwakernaak2023counting,waheed2020boolean,zhang2023mechanical}, and soft pneumatic circuits~\cite{weaver2010static,mosadegh2011next,preston2019digital,song2021cmos,stanley2024high,drotman2021electronics}.
These approaches have demonstrated that information processing can be realized in mechanical deformation.

In most mechanical computers, the logic function is linked to the geometry of the materials. The operation of the components depends on geometric parameters such as thicknesses, stiffnesses, or shapes, which must be chosen appropriately to achieve the desired behavior. As a result, variations introduced during fabrication or deformation can modify the mechanical response of the system and affect its function. Although many designs achieve good robustness within their intended operating range, maintaining reliable performance often requires careful control of geometry and material properties. These considerations motivate us for computational architectures whose logic is defined independently of details of geometry.

In the knitted computer, the logic and its physical execution are separated. The propagation rules are determined by the topology, and therefore remain unchanged under deformation of the textile. By contrast, mechanical parameters such as tension, friction, and yarn properties influence whether unraveling can propagate through the fabric and complete the computation. As demonstrated in the experiments, insufficient force may stop the propagation, but it does not generate wrong outputs. This distinction is a defining feature of the knitted computer. Topology specifies the logic, whereas mechanics provides the physical system by which it is executed.

\section{conclusion}

We implemented information processing on a fabric through chains of unraveling and introduced a mechanical computer in which logic operations are defined by topology rather than than geometry. The knitted fabric performs computation through the controlled unraveling of a single continuous yarn, where each bit is represented by a pair of stitches and logic functions are determined by the interaction among the stitches.

This work establishes a fundamental distinction between logic, determined by topology, and mechanics, which executes the process. This separation makes the computation robust against geometric and material variations and suggests a topological mechanical computer. This framework may be extended to systems across a wide range of length scales, including polymers, textiles, and architectural structures.

\section{Methods}
\label{physical}

The logic elements were knitted from a single yarn composed of a natural rubber cord covered with rayon fiber (KW91325, Kawamura Seichu Co., diameter 1~mm).
The yarn was knitted manually into a textile with a desired topology. Four wooden rods (diameter 3 or 5~mm) were inserted to the edges, and a rectangular frame was formed. The frame made it easy to drag the edges to apply tension to the fabric.
Typical circuit samples contained $6$-$10$ stitches in width and $6$-$20$ rows in height, depending on the number of gates implemented.

Inputs were set according to the two-cell convention described in Sec.~\ref{propagation_sets_for_logic}.
For each input bit, one of the two adjacent stitches was released to set its state to $U$, while the other remained in state $K$.
Before computation, all other stitches were knitted and kept as state $K$.
This configuration fixed the initial condition $\sigma^{(0)}$ corresponding to the desired inputs. 

In the experiment, the fabric was stretched to apply tension by pulling the vertical rods outward using a linear guide and a stepping motor with the horizontal rods fixed. Then, some designated stitches were released to set the input value. The computation proceeds immediately after setting the input. The stress was measured by load cell (LUR-A-100NSA1, Kyowa) attatched to the vertical rods.

\begin{acknowledgments}
This work benefited from fruitful discussions with, but not limited to S. Poinloux. The author thanks the use of the facilities at the Wada Laboratory, Ritsumeikan University, and acknowledges the financial support provided by JSPS KAKENHI Grant-in-Aid for JSPS Fellows 23KJ0753 and 26KJ0356. 
\end{acknowledgments}

\bibliography{knit}

\appendix

\begin{figure}
    \centering
    \includegraphics[width=1.0\linewidth,bb=0 0 260 200]{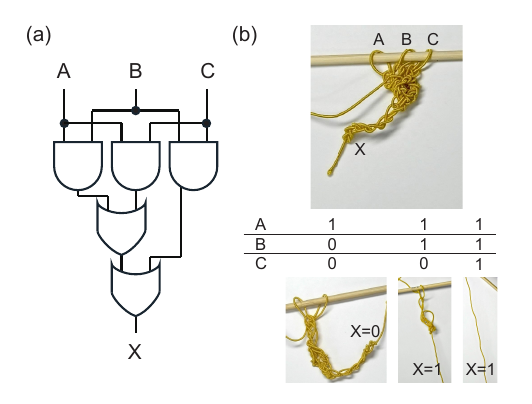}
    \caption{Circuit diagram (a) and experimental realization (b) of the three-input majority function. It was confirmed that the output is 1 only when more than one input is 1.}
    \label{monotone}
\end{figure}

\section{Monotone logic and compact design}
\label{monotone_section}

The knitted computer developed here can make monotone Boolean functions, which include only AND and OR without NOT gate, in the easier way than the universal logic circuit.  
In this representation, the knitted and unknitted states of one stitch can be directly assigned to logical 0 and 1, respectively, eliminating the need for the two-cell encoding used for universal logic.  
Each cell then represents a single bit, and the product and sum modules implement AND and OR through their topological connectivity.  

This monotone representation sacrifices universality of logic function and neutral state, but enables more compact circuit layouts.  
Because each bit corresponds to a single cell, the overall fabric size and complexity can be reduced.  
Experimentally, we verified that monotone circuits such as the three-input majority function
\begin{equation}
    X=(A\land B)\lor (B\land C)\lor (C\land A)
\end{equation}
works correctly under this simplified way.  
The result confirms that the computation can implement monotone Boolean functions.

\end{document}